%% file: main.tex
\documentclass[11pt]{article}
\usepackage{style}
\usepackage{shortcuts}

\title{A Simple Parallel Algorithm with Near-Linear Work\\for Negative-Weight Single-Source Shortest Paths}
\author{
    Nick Fischer\thanks{INSAIT, Sofia University ``St. Kliment Ohridski''. This research was partially funded from the Ministry of Education and Science of Bulgaria (support for INSAIT, part of the Bulgarian National Roadmap for Research Infrastructure). Parts of this work were done while the author was affiliated with Weizmann Institute of Science and visiting ETH Zurich. The visit was funded through the European Research Council (ERC) under the European Union's Horizon 2020 research and innovation program (ERC grant agreement 949272).} \and
    Bernhard Haeupler\thanks{INSAIT, Sofia University ``St. Kliment Ohridski'' and ETH Zurich.  Partially funded from the Ministry of Education and Science of Bulgaria (support for INSAIT, part of the Bulgarian National Roadmap for Research Infrastructure).  Also partially funded through the European Research Council (ERC) under the European Union's Horizon 2020 research and innovation program (ERC grant agreement 949272). } \and
    Rustam Latypov\thanks{Aalto University. Supported by the Research Council of Finland, Grant 334238.} \and
    Antti Roeyskoe\thanks{ETH Zurich. Partially funded through the European Research Council (ERC) under the European Union's Horizon 2020 research and innovation program (ERC grant agreement 949272). Parts of this work were done while the author was visiting INSAIT.
    Partially funded from the Ministry of Education and Science of Bulgaria (support for INSAIT, part of the Bulgarian National Roadmap for Research Infrastructure).} \and
    Aurelio L. Sulser\thanks{ETH Zurich.}}
\date{}

\begin{document}

\maketitle

\begin{abstract}
\noindent
We give the first parallel algorithm with optimal $\tilde{O}(m)$ work for the classical problem of computing Single-Source Shortest Paths in general graphs with negative-weight edges. 

\medskip
In graphs without negative edges, Dijkstra's algorithm solves the Single-Source Shortest Paths (SSSP) problem with optimal $\tilde O(m)$ work, but is inherently sequential. A recent breakthrough by [Bernstein, Nanongkai, Wulff-Nilsen; FOCS '22] achieves the same for general graphs. Parallel shortest path algorithms are more difficult and have been intensely studied for decades. Only very recently, multiple lines of research culminated in parallel algorithms with optimal work $\tilde O(m)$ for various restricted settings, such as approximate or exact algorithms for directed or undirected graphs without negative edges. For general graphs, the best known algorithm by [Ashvinkumar, Bernstein, Cao, Grunau, Haeupler, Jiang, Nanongkai, Su; ESA '24] still requires $m^{1+o(1)}$ work.

\medskip
This paper presents a randomized parallel algorithm for SSSP in general graphs with near-linear work $\tilde O(m)$ and state-of-the-art span $n^{1/2 + o(1)}$. We follow a novel bottom-up approach leading to a particularly clean and simple algorithm. Our algorithm can be seen as a \emph{near-optimal parallel black-box reduction} from SSSP in general graphs to graphs without negative edges. In contrast to prior works, the reduction in this paper is both parallel and essentially without overhead, only affecting work and span by polylogarithmic factors. 
\end{abstract}

\thispagestyle{empty}
\setcounter{page}{0}
\clearpage

\input{sections/intro}
\input{sections/prelims}
\input{sections/algo}

\bibliographystyle{plainurl}
\bibliography{refs}

\end{document}

%% file: sections/intro.tex
\section{Introduction}
This paper provides a parallel algorithm with $\tilde{O}(m)$ work and state-of-the-art span $n^{1/2 + o(1)}$ for the classical SSSP problem in general graphs with negative-weight edges.

The \emph{Single-Source Shortest Paths (SSSP)} problem is one of the most fundamental problems in algorithmic graph theory, and among the most difficult and well-studied problems in parallel computing. The input consists of a directed graph with $n$ nodes, $m$ edges and integer edge weights (which in general may be negative), and a designated source vertex. The goal is to compute a shortest-path tree from the source to all other vertices (or find a negative cycle if one exists).

Without negative edges, shortest paths can be computed in near-linear time $O(m + n \log n) = \tilde{O}(m)$ by the classical algorithm of Dijkstra~\cite{Dijkstra59,Williams64,FredmanT84}. Unfortunately Dijkstra's algorithm is inherently \emph{sequential} and inherently requires \emph{nonnegative} edge weights. Lifting these constraints poses two algorithmic barriers that have been the subject of intense study in recent years.

On the one hand, achieving any \emph{parallel} algorithm has been a major open problem, even for unweighted, undirected graphs (i.e, even for simply computing a Breadth First Search tree). Indeed, no parallel shortest path or breadth first search algorithm with sub-linear $n^{1-\Theta(1)}$ span without a polynomial increase in work was known until 2023~\cite{RozhovnHMGZ23,CaoF23}. While exact parallel algorithms for general graphs remained out of reach until recently, parallel algorithms for (approximate) shortest paths in (undirected) graphs without negative edges have been intensely studied for decades and have been a driving force for discovery of new algorithmic techniques. Here, we give a short list of highlights among parallel shortest path algorithms over the last three decades, restricted to algorithms achieving at least roughly linear work and a state-of-the-art span, that is, polylogarithmic for undirected graphs and $n^{1/2 + o(1)}$ for directed graphs:
\begin{itemize}[itemsep=0pt]
    \item the '94 breakthrough of Cohen~\cite{Cohen00}, which gave a parallel algorithm with $O(m^{1+\epsilon})$ work for approximations in undirected graphs with no negative-weight edges, 
    \item the algorithms of Andoni, Stein and Zhong~\cite{AndoniSZ20} and Li~\cite{Li20} and Rozhon, Grunau, Haeupler, Zuzic and Li~\cite{RohzhvnGHZL22} improving this to $\tilde{O}(m)$ work for the same setting,
    \item another series of celebrated works by Fineman for directed graphs, which cumulated in the approximation algorithm for directed graphs without negative edges of Cao, Fineman and Russell~\cite{CaoFR20},
    \item and lastly, improvements to exact algorithms for directed graphs without negative edges by Rozhon, Haeupler, Martinsson, Grunau and Zuzic~\cite{RozhovnHMGZ23} and Cao and Fineman~\cite{CaoF23}.
\end{itemize}

On the other hand, the design of sequential algorithms that can deal with \emph{negative} edge weights has been equally challenging with research progress having stagnated at a complexity of~$\tilde{O}(m \sqrt{n})$ since Goldberg's 1993 scaling framework~\cite{Gabow83,GabowT89,Goldberg95}. Only in 2022, a breakthrough by Bernstein, Nanongkai, and Wulff-Nilsen~\cite{BernsteinNW22} resulted in a sequential shortest-path algorithm with near-optimal running time $\tilde{O}(m)$ (see also~\cite{BringmannCF23}). At the same time, independently,~\cite{ChenKLPGS22} devised an almost-optimal $m^{1+\order(1)}$ for Minimum-Cost Flow and thereby also for Negative-Weight SSSP. Notably, in another breakthrough Fineman~\cite{Fineman24,HuangJQ24} recently showed that there is a \emph{strongly polynomial} $\tilde\Order(m n^{8/9})$-time algorithm.

At the intersection of these two mostly disjoint lines of research, it is known that Goldberg's scaling algorithm can be ported to the parallel setting with work $\tilde\Order(m \sqrt n)$ and span~$n^{5/4+o(1)}$~\cite{CaoFR22}. Only very recently, Ashvinkumar, Bernstein, Cao, Grunau, Haeupler, Jiang, Nanongkai and Su~\cite{AshvinkumarBCGHJNS24} devised a parallel algorithm for Negative-Weight SSSP based on the breakthrough by Bernstein et al.~\cite{BernsteinNW22}. Their algorithm increases the total computational work to $m^{1+o(1)}$ and achieves span $n^{1/2+o(1)}$. Achieving any parallelism (say, any sublinear span) with near-optimal~$\tilde{O}(m)$ work was left as an open problem.

\subsection{Our Contribution}

In this work, we provide the first parallel algorithm for Negative-Weight Single-Source Shortest Paths with near-optimal $\tilde{O}(m)$ work and state-of-the-art span.

\begin{theorem}[Parallel Negative-Weight Single-Source Shortest Paths] \label{thm:parallel-nsssp}
There is a randomized parallel algorithm for Negative-Weight Single-Source Shortest Paths with work $\tilde{O}(m)$ and span $n^{1/2+o(1)}$.\footnote{Here and throughout we assume that all weights are polynomially bounded (i.e., that there is some constant~$c$ such that all weights are in $\{-n^c, \dots, n^c\}$). Weights from a larger range $\{-M, \dots, M\}$ only lead to a slowdown of $O(\log M)$.}
\end{theorem}

We obtain the result through a conceptually new way of thinking about the Bernstein-Nanongkai-Wulff-Nilsen algorithm in a bottom-up rather than a top-down fashion. This bottom-up approach leads to a much simpler parallel  algorithm.

Similar to prior works~\cite{BernsteinNW22,BringmannCF23,AshvinkumarBCGHJNS24} our algorithm can be seen as a black-box reduction from the SSSP problem in general graphs to graphs without negative edges. In contrast to prior works, the reduction in this paper is both parallel and essentially optimal, i.e., the overhead for both work and span are merely polylogarithmic factors. 

\begin{theorem}[Black-Box Reduction]\label{thm:parallel-nsssp-reduction}
If there is a parallel Nonnegative-Weight SSSP algorithm with work $W_{\geq 0}(n, m)$ and span $S_{\geq 0}(n, m)$, then there is a randomized parallel Negative-Weight SSSP algorithm with work $W_{\geq 0}(n, m) \cdot \poly(\log n)$ and span $S_{\geq 0}(n, m) \cdot \poly(\log n)$. 
\end{theorem} 

We obtain \Cref{thm:parallel-nsssp} by combining this reduction with the best known algorithms for parallel Nonnegative-Weight SSSP~\cite{CaoF23,RozhovnHMGZ23} with work $\tilde{O}(m)$ and span~$n^{1/2 + o(1)}$. Our reduction ensures that any future progress on the parallel span for Nonnegative-Weight SSSP (even for approximate shortest paths, due to the black-box reduction provided in~\cite{RozhovnHMGZ23}) immediately translates to the same improvements for exact algorithms in general graphs. Like~\cite{AshvinkumarBCGHJNS24} our algorithm work in the very general minor-aggregation framework of~\cite{ZuzicGYHS21,RohzhvnGHZL22,GhaffariZ22} which allows designing algorithms that simultaneously apply to parallel, distributed and other settings in a model-independent way.

Besides its applications to parallel computing, we are confident that our conceptual insights will lead to further progress on the Negative-Weight SSSP problem, particularly deterministic sequential and parallel algorithms, as well as algorithms for distributed and other models of computation.

\paragraph{Paper Structure.}
In \Cref{sec:prelims}, we give an overview of the techniques developed in prior work on the parallel negative Single-Source Shortest Paths problem. In \Cref{sec:main-result}, we present our new bottom-up approach and prove \Cref{thm:parallel-nsssp-reduction}.

%% file: sections/prelims.tex
\section{Preliminaries}\label{sec:prelims}
We write $[n] = \set{1, \dots, n}$ and set $\poly(n) = n^{\Order(1)}$ and $\tilde\Order(n) = n (\log n)^{\Order(1)}$.

Throughout the paper, we consider the standard PRAM model for parallel computation.

Following standard terminology, we say that an event happens \emph{with high probability} if it happens with probability at least $1 - 1/n^c$, where $c$ is an arbitrarily large prespecified constant (and $n$ is the input size). A randomized algorithm is \emph{Monte Carlo} if it succeeds with high probability (and may err with low probability), and we say that a randomized algorithm is \emph{Las Vegas} if it succeeds with probability~$1$.

\subsection{Negative-Weight SSSP}
We denote directed graphs by $G = (V, E, w)$ with vertex sets $V$, directed edge sets $E \subseteq V^2$ and (potentially negative) integer edge weights $w : E \to \mathbb{Z}$. If $G$ does not contain negative-length cycles, we write $\dist_G(v, u)$ to denote the length of a shortest $v$-$u$-path. For $S \subseteq V$, we say that~$P$ is a \emph{shortest $S$-$v$-path} if~$P$ has smallest length among all paths starting in some node in~$S$ and ending in $v$. We also set $\dist_G(S, v) = \min_{u \in S} \dist_G(u, v)$. (We often drop the subscript $G$ when it is clear from context.) Finally, we say that a set $S \subseteq V$ has \emph{weak diameter $d$} in $G$ if for all $u, v \in S$ it holds that $\dist_G(u, v) \leq d$ and $\dist_G(v, u) \leq d$.

\begin{definition}[SSSP]
The \emph{Single-Source Shortest Paths (SSSP)} problem is, given a directed graph $G = (V, E, w)$ (with possibly negative integer edge weights $w$) and a source vertex $s \in V$, to
\begin{itemize}
    \item output a shortest path tree from $s$, or
    \item output a negative cycle in $G$.
\end{itemize}
\end{definition}

The \emph{Nonnegative-Weight SSSP} problem is the natural restriction to graphs with nonnegative edge lengths. Throughout, let us denote by~$W_{\geq 0}(n, m)$ and $S_{\geq 0}(n, m)$ the work and span of any parallel algorithm for nonnegative SSSP. We make the reasonable assumption that 
\begin{align}\label{eq:work-span-additive}
\begin{split}
    W_{\geq 0}(n_1, m_1) + W_{\geq 0}(n_2, m_2) &\leq W_{\geq 0}(n_1 + n_2, m_1 + m_2), \\
    \max(\,S_{\geq 0}(n_1, m_1),\, S_{\geq 0}(n_2, m_2)\,) &\leq S_{\geq 0}(n_1 + n_2, m_1 + m_2).
\end{split}
\end{align}
The currently best-known algorithms for Nonnegative-Weight SSSP achieve $W_{\geq 0}(n,m) = \tilde{O}(m)$ and $S_{\geq 0}(n, m) = n^{1/2 + o(1)}$~\cite{CaoF23,RozhovnHMGZ23}.

Throughout, for the sake of simplicity we assume that the edge weights in the given graph are polynomially bounded (i.e., there is some constant $c$ such that all edge weights satisfy $|w(e)| \leq n^c$). We emphasize that this restriction is indeed merely for simplicity, and that larger edge weights up to $M$ only lead to an overhead of $\log M$ in the work and span of our algorithm.

\subsection{Directed Low-Diameter Decompositions}
A crucial ingredient and one of the main innovations of~\cite{BernsteinNW22} is the use of a Low-Diameter Decomposition (LDD) in \emph{directed} graphs. Specifically, we rely on the following lemma of~\cite{AshvinkumarBCGHJNS24} which already ported directed LDDs to the parallel setting:

\begin{theorem}[Directed LDD~{{\cite[Theorem 1.4]{AshvinkumarBCGHJNS24}}}] \label{thm:dirldd}
    Let $G = (V, E, w)$ be a directed graph with nonnegative edge weights $w : E \rightarrow \mathbb{Z}_{\geq 0}$ and let $d \geq 1$. There is a randomized algorithm computing a set of edges $E^{\mathrm{rem}} \subseteq E$ with the following guarantees:
    \begin{itemize}
        \item Each strongly connected component of the graph $E \setminus E^{\mathrm{rem}}$ has weak diameter at most $d$ in $G$, i.e. if $u, v$ are two vertices in the same SCC, then $\dist_G(u, v) \leq d$ and $\dist_G(v, u) \leq d$.
        \item For any $e \in E$, we have $\prob(e \in E^{\mathrm{rem}}) = \bigO\left(\frac{w(e) \log^2 n}{d} + \frac{1}{n^8}\right)$.
    \end{itemize}
    The algorithm has work $W_{\geq 0}(n, m) \cdot \poly(\log n)$ and span $S_{\geq 0}(n, m) \cdot \poly(\log n)$ span (assuming that $w$ is polynomially bounded).
\end{theorem}

In what follows it will be convenient to additionally assume that the strongly connected components are topologically sorted. For this reason, we combine \cref{thm:dirldd} with a parallel algorithm for computing topological orderings (e.g.,~\cite[Lemma 5.1]{AshvinkumarBCGHJNS24} with work $W_{\geq 0}(n, m)$ and span $S_{\geq 0}(n, m)$ up to polylogarithmic overhead), to immediately obtain the following lemma:

\begin{lemma}[Topologically Sorted LDD]\label{cor:topo-sorted-directed-ldd}
Let $G = (V, E, w)$ be a directed graph with nonnegative edge weights $w : E \to \Int_{\geq 0}$, and let~\makebox{$d \geq 1$}. There is a randomized algorithm $\textsc{DirLDD}(G, d)$ that computes a partition $\calS = (S_1, \dots, S_k)$ of~$V$ satisfying the following two properties:
\begin{itemize}
    \item For all $i \in [k]$ and for all vertices $u, v \in S_i$, we have $\dist_G(u, v) \leq d$ and $\dist_G(v, u) \leq d$.
    \item We say that an edge $e = (u, v) \in E$ is \emph{cut} if, for the unique indices $i, j$ with $u \in S_i, v \in S_j$, it holds that $i > j$. The probability that $e$ is cut is at most~\smash{$\lddqual(n) \cdot \frac{w(e)}{d} + \bigO(\frac{1}{n^8})$}, where $Q(n) = \Order(\log^2 n)$.
\end{itemize}
The algorithm has work $W_{\geq 0}(n, m) \cdot \poly(\log n)$ and span $S_{\geq 0}(n, m) \cdot \poly(\log n)$ (assuming that~$w$ is polynomially bounded).
\end{lemma}

We will make the assumption that \cref{cor:topo-sorted-directed-ldd} does not cut any edge in the case where all pairs of nodes that can reach each other are within distance $d$ (i.e., in the trivial instances in which we only need to compute a topological ordering of the strongly connected components).

\subsection{Potential Adjustments}
A key tool used in the context of shortest paths are potential adjustments, originally introduced by Johnson~\cite{Johnson77}.

\begin{definition}[Potential-Adjusted Weights]
Let $G = (V, E, w)$ be a directed graph and let $\phi : V \rightarrow \mathbb{Z}$ be a potential function. The \textit{$\phi$-adjusted edge weights} $w_\phi$ are defined by $w_\phi(u, v) = w(u, v) + \phi(u) - \phi(v)$. We write $G_\phi = (V, E, w_\phi)$.
\end{definition}

It is easy to observe that potential adjustments preserve shortest paths (in the sense that~$P$ is a shortest $u$-$v$-path in $G$ if and only if $P$ is a shortest $u$-$v$-path in $G_\phi$ for all potentials $\phi$), and preserve the weight of cycles exactly. Another useful application is \emph{Johnson's trick:}

\begin{observation}[Johnson's Trick~\cite{Johnson77}]\label{obs:johnson}
Let $G=(V, E)$ be a directed graph with no negative cycles, and let $\phi$ be the potential defined by $\phi(v) = \dist(V, v)$. Then all edge weights in $G_\phi$ are nonnegative.
\end{observation}

Johnson's trick is a tool to find a potential that makes all edges in the graph nonnegative; the remaining instance of Nonnegative-Weight SSSP can then be solved e.g.\ by Dijkstra's algorithm. In our case, however, we can only find a potential that makes sure that the shortest paths contain \emph{few} negative edges. In this case, we rely on the following layered construction (\cref{lem:dijkstra-potential}) to nevertheless reduce to the nonnegative case. Notably, this simple lemma replaces the necessity to invoke the interwoven ``Dijkstra-Bellman-Ford'' algorithm as in~\cite{BernsteinNW22,BringmannCF23}.

\begin{lemma}[SSSP with Few Negative Edges] \label{lem:dijkstra-potential}
Let $G = (V, E, w)$ be a directed graph without negative weight cycles, let $\phi : V \to \Int$ and $t \geq 0$. There in an algorithm $\textsc{FewNegSSSP}(G, \phi, t)$ that returns $\boldsymbol{d} : V \to \Int$ satisfying
\begin{itemize}
    \item $\dist(V, v) \leq \boldsymbol{d}(v)$, and
    \item $\boldsymbol{d}(v) \leq w(P)$ for every path $P$ ending in $v$ containing at most $t$ negative edges with respect to $w_\phi$ (i.e. $|\{e \in P : w_\phi(e) < 0\}| \leq t$).
\end{itemize}
The algorithm has work $\Order(W_{\geq 0}(nt, mt))$ and span $\Order(S_{\geq 0}(nt, mt))$ (assuming that both $w$ and~$\phi$ are polynomially bounded).
\end{lemma}
\begin{proof}
Consider the $(t+1)$-layered graph $G' = (V', E', w')$ with vertices $V' = V \times \set{0, \dots, t}$ and edges defined as follows: For each original edge $(u, v) \in E$ and for each $0 \leq i \leq t$ we include an edge $((u, i), (v, i)) \in E'$ of weight $w_\phi(u, v)$ if this weight happens to be nonnegative. Moreover, for each $0 \leq i < t$ we include the edge $((u, i), (v, i+1)) \in E'$ of weight $w_\phi(u, v) + M$, where we choose a sufficiently large value $M := 2 \max_v |\phi(v)| + \max_v |w(v)|$ ensuring that the edge weight becomes nonnegative. Additionally, attach a source vertex $s$ to the graph~$G'$ that is connected to each node $(u, 0)$ via an edge of length $-\phi(u) + M$. All in all, the graph $G'$ contains only nonnegative edges, and thus we can afford to compute the distances $\dist_{G'}(s, v')$ for all $v' \in V'$ with work~$\Order(W_{\geq 0}(nt, mt))$ and span $\Order(S_{\geq 0}(nt, mt))$. Finally, we return $\boldsymbol{d}(v) = \dist_{G'}(s, (v, t)) - (t+1)M + \phi(v)$.

For the correctness, note that any $s$-$(v, t)$-path $P'$ in $G'$ traverses the layers~\makebox{$0, \dots, t$} (in this order). Inside every layer the path arbitrarily follows the edges from $G$, and between the layers the path takes exactly one edge from $G$. Formally, the path takes the form
\begin{equation*}
    P' = s \, (v_{0, 1}, 0) \, \dots \, (v_{0, \ell_0}, 0) \, (v_{1, 1}, 1) \, \dots \, (v_{t-1, \ell_{t-1}}, t-1) \, (v_{t, 1}, t) \, \dots \, (v_{t, \ell_t}, t) = (v, t).
\end{equation*}
But by construction, $P = v_{0, 1} \, \dots \, v_{0, \ell_0} \, v_{1, 1} \, \dots \, v_{t-1, \ell_{t-1}} \, v_{t, 1} \, \dots \, v_{t, \ell_t}$ is a path in the original graph~$G$, and it holds that $w'(P') = w(P) + (t+1) M - \phi(v)$ (since the contribution of the potentials telescope, and since $P'$ contains exactly $t+1$ edges that incur an overhead of $+M$). It follows immediately that $\boldsymbol{d}(v) \geq \dist_G(V, v) + (t+1) M - \phi(v) - (t+1) M + \phi(v) = \dist_G(V, v)$ proving the first item. For the second item, observe that any path $P$ to $v$ that contains at most~$t$ negative edges with respect to $w_\phi$ can be simulated as an $s$-$(v, t)$-path in $G'$ by traversing the~$t$ negative edges \emph{between} layers.
\end{proof}

\subsection{Reduction to Restricted Graphs}\label{sec:restricted-graphs}
Consider the following \emph{restricted} variant of SSSP:

\begin{definition} \label{def:restricted-graph}
An edge-weighted directed graph $G=(V, E, w)$ is \emph{restricted} if it satisfies the following two conditions:
\begin{enumerate}[label=(\roman*), topsep=\smallskipamount]
    \item all edge weights are at least $-1$ and at most $n$ (i.e., $w(e) \in \set{-1, 0, \dots, n}$), and
    \item the minimum cycle mean is at least 1 (i.e. for every directed cycle $C$, $\sum_{e \in C} w(e) \geq |C|$).
\end{enumerate}
\end{definition}

\begin{definition}[Restricted SSSP] \label{def:restricted-sssp}
Given a restricted graph $G = (V, E, w)$, the goal is to compute a potential $\phi : V \to \Int$ such that $w_\phi(e) \geq 0$ for all $e \in E$.
\end{definition}

In particular, note that restricted graphs cannot contain negative cycles, and thus in the Restricted SSSP problem all distances are well-defined. Already in~\cite{BernsteinNW22} a scaling technique was formalized to reduce from the general SSSP problem to Restricted SSSP; this argument was later made explicit and refined in~\cite{BringmannCF23}. Phrased in terms of parallel algorithms, the reduction can be summarized as follows:\footnote{Specifically, \cite[Section 5]{BringmannCF23} outlines a reduction from general SSSP to Restricted SSSP which calls the Restricted SSSP oracle only a polylogarithmic number of times. It is trivially to parallelize this reduction simply by making these calls sequential (at the cost of increasing work and span by a polylogarithmic factor). In~\cite{BringmannCF23}, however, the definition of \emph{restricted} differs slightly from ours in that it does not impose the upper bound $w(e) \leq n$. To enforce this additional constraint, we simply remove all edges with weight more than $n$ in the constructed graph $G$; let~$G'$ denote the remaining graph (which is restricted according to \cref{def:restricted-graph}). Solving Restricted SSSP on $G'$ leads to a potential that makes all edges in $G'$ nonnegative, and thus by one additional call to Nonnegative-Weight SSSP we can compute the shortest $V$-$v$-paths in $G'$. But we have that $\dist_G(V, v) = \dist_{G'}(V, v)$ as no edge with weight more than $n$ can participate in a shortest $V$-$v$-path (given that all edges have weight at least $-1$).}

\begin{lemma}[\cite{BernsteinNW22,BringmannCF23}] \label{lem:restricted-suffices}
Suppose that there is a Monte Carlo algorithm for Restricted SSSP with work $W_{\mathrm{R}}(n, m)$ and span $S_{\mathrm{R}}(n, m)$. Then there is a Las Vegas algorithm for SSSP with work $(W_{\mathrm{R}}(n, m) + W_{\geq 0}(n, m)) \cdot \poly(\log n)$ and span $(S_{\mathrm{R}}(n, m) + S_{\geq 0}(n, m)) \cdot \poly(\log n)$.
\end{lemma}

%% file: sections/algo.tex
\section{A Bottom-Up Algorithm for Negative-Weight SSSP}\label{sec:main-result}
In this section we give an efficient parallel algorithm for Restricted SSSP; see the pseudocode \textsc{RestrictedSSSP} in \cref{alg:restricted-sssp}. Our goal is to establish the following \cref{thm:restricted-sssp}. Combined with \Cref{lem:restricted-suffices}, this establishes our Main \Cref{thm:parallel-nsssp-reduction}.

\begin{restatable}{theorem}{meancycleworks}\label{thm:restricted-sssp}
    For a restricted graph $G = (V, E, w)$, the algorithm $\textsc{RestrictedSSSP}(G)$ computes a potential $\phi : V \rightarrow \mathbb{Z}$ such that, with high probability, $w_\phi(e) \geq 0$ for all edges $e \in E$. The algorithm has work $W_{\geq 0}(n, m) \cdot \poly(\log n)$ and span $S_{\geq 0}(n, m) \cdot \poly(\log n)$.
\end{restatable}

\begin{algorithm}[t]
    \caption{A parallel algorithm for the Restricted SSSP problem; see \cref{thm:restricted-sssp}.} \label{alg:restricted-sssp}
    \begin{algorithmic}[1]
        \Function{RestrictedSSSP}{$G$}
            \State Let $L := \ceil{\log(2n^2)}$
            \State Let $R := (2 + c_{\mathrm{whp}}) \log n$ \Comment{$c_{\mathrm{whp}}$ controls probability of success}
            \State Let $t := 4\lddqual(n) + 1$ \Comment{$Q(n)$ is the quality of \Cref{cor:topo-sorted-directed-ldd}}
            \State Let $\Phi_0 \gets \set{0}$ be the singleton set consisting of the all-zero potential
            \For{$\ell \gets 1, \dots, L$}
                \For{$r \gets 1, \dots, R$}
                    \State Let $\mathcal S_{\ell, r} = (S_1, \dots, S_k) \gets \Call{DirLDD}{G_{\geq 0}, 2^\ell}$
                    \ForParallel{$i \in [k]$}
                        \For{$\phi \in \Phi_{\ell-1}$}
                            \State Let $\boldsymbol{d}_\phi \gets \Call{FewNegSSSP}{G[S_i], \phi, t}$
                        \EndFor
                        \State Let $\phi^{(i)}$ be the potential defined $\phi^{(i)}(v) = \min_{\phi \in \Phi_{\ell-1}} \boldsymbol{d}_\phi(v)$
                    \EndForParallel
                    \State Let $\phi_{\ell, r}$ be the potential defined by~\smash{$\phi_{\ell, r}(v) = \phi^{(i)}(v) - i \cdot n$} whenever $v \in S_i$
                \EndFor
                \State Let $\Phi_\ell \gets \set{\phi_{\ell, r} : r \in [R]}$
            \EndFor
            \State \Return $\phi_{L, 1}$
        \EndFunction
    \end{algorithmic}
\end{algorithm}

\subsection{Overview}
We start with a brief intuitive description of the algorithm. Throughout let $G_{\geq 0} = (V, E, w_{\geq 0})$ be the graph $G$ where we replace all negative edge weights by zero (i.e. $w_{\geq 0}(e) = \max(0, w(e))$). The algorithm proceeds in a logarithmic number of \emph{levels} $\ell \gets 1, \dots, L$ \emph{from bottom to top}. At each level $\ell$ we further run a logarithmic number of \emph{repetitions} $r \gets 1, \dots, R$, and for each repetition we compute a low-diameter decomposition $\mathcal S_{\ell, r}$ in the graph $G_{\geq 0}$ with diameter parameter $d = 2^\ell$. Our goal is to find a potential function $\phi_{\ell, r}$ that reweights the graph in such a way that all edges not cut by~$\mathcal S_{\ell, r}$ become nonnegative. At the top-level $L$ we do not cut any edges, and thus the function~$\phi_{L, 1}$ (say) makes all edges nonnegative. To construct the potential~$\phi_{\ell, r}$ we follow these two steps:
\begin{enumerate}
    \setlength\parindent{1.6em}
    \item We first consider each cluster $S_i$ in the LDD $\mathcal S_{\ell, r}$ individually. Our goal is to compute a potential function $\phi^{(i)}$ that makes all edges \emph{inside} $S_i$ nonnegative. Inspired by Johnson's trick, our goal is to assign $\phi^{(i)}(v) = \dist_{G[S_i]}(S_i, v)$, i.e., to compute the shortest $S_i$-$v$-paths for all $v \in S_i$.
    
    Fix a shortest $S_i$-$v$-path $P$ in $G[S_i]$. A key insight of Bernstein, Nanongkai and Wulff-Nilsen is that since $G$ is restricted and since $S_i$ has weak diameter at most $d = 2^\ell$ in~$G_{\geq 0}$, the path $P$ has small length even in the graph $G_{\geq 0}$; specifically,~\makebox{$w_{\geq 0}(P) \leq d$} (see \cref{lem:rep-paths-short}). Therefore, with good probability the path $P$ has been cut only a few times (say, $t = \poly(\log n)$ times) on the \emph{previous level $\ell-1$}. Therefore, on the previous level we have already computed a potential function $\phi_{\ell-1, r'}$ (in some repetition $r'$) which makes all but at most~$t$ edges in~$P$ nonnegative. We can thus compute $w(P) = \dist_{G[S_i]}(S_i, v)$ by calling $\Call{FewNegSSSP}{G[S_i], \phi_{\ell-1, r'}, t}$ (see \cref{lem:dijkstra-potential}).

    \item In a second step, we combine the potential functions~\smash{$\phi^{(i)}$} of the individual clusters to one potential function $\phi_{\ell, r}$. The idea is simple: We take~\smash{$\phi_{\ell, r}(v) = \phi^{(i)}(v) - i \cdot M$} whenever~\makebox{$v \in S_i$}, for some sufficiently large number $M$. The potentials $\phi^{(i)}$ ensure that edges within the same cluster are nonnegative, and any edge from a lower-index cluster to a higher-index cluster is nonnegative provided that $-M < \phi^{(i)}(v) \leq 0$ (it turns out that taking $M = n$ suffices). 
\end{enumerate}

\subsection{Formal Analysis}
In this section we provide a formal analysis of \cref{alg:restricted-sssp}. Our strategy is as follows: We first recap one of the key insights of~\cite{BernsteinNW22} phrased in our language (\cref{lem:rep-paths-short}). We then define a certain \emph{success event} (\cref{def:success-event}), and we prove that this event happens with high probability (\cref{lem:success-event}). Conditioned on this success event we can then prove by induction from bottom to top that the algorithm is correct (\cref{lem:restricted-sssp-correctness}), and efficient (\cref{lem:restricted-sssp-work}).

\begin{lemma} \label{lem:rep-paths-short}
Let $G = (V, E, w)$ be a restricted graph and let $S \subseteq V$ be such that $S$ has weak diameter $d$ in $G_{\geq 0}$. Then any shortest $S$-$v$-path $P$ in $G[S]$ has length $|P| \leq d$ and weight $w_{\geq 0}(P) \leq d$.
\end{lemma}
\begin{proof}
Let $u \in S$ denote the starting vertex of~$P$, let $Q$ denote a shortest $v$-$u$-path in $G[S]$ and let $C$ denote the cycle obtained by concatenating~$P$ and $Q$. On the one hand, we have that
\begin{equation*}
    w(C) = w(P) + w(Q) \leq w(Q) \leq w_{\geq 0}(Q) \leq \dist_{G_{\geq 0}}(v, u) \leq d,
\end{equation*}
where the first inequality is due to~\smash{$w(P) \leq \dist_{G[S]}(S, v) \leq 0$} and the last inequality is due to the fact that $S$ has weak diameter at most $d$ in $G_{\geq 0}$. On the other hand, since~$G$ is a restricted graph, we have that $|C| \leq w(C)$. Putting both together, we obtain that
\begin{equation*}
    w_{\geq 0}(P) \leq w(P) + |P| \leq |P| \leq |C| \leq w(C) \leq d,
\end{equation*}
proving the statement.
\end{proof}

\begin{definition}[Success Event] \label{def:success-event}
The algorithm is \emph{successful} if for all levels $1 < \ell \leq L$, for all $r \in [R]$ and for all vertices $v \in V$ the following holds: Let $S \in \mathcal S_{\ell, r}$ be the unique cluster containing $v$. Then there is a shortest $S$-$v$-path $P$ in $G[S]$ (i.e., $w(P) = \dist_{G[S]}(S, v)$) and a repetition $r' \in [R]$ such that $\mathcal S_{\ell-1, r'}$ cuts at most $t$ edges from $P$.
\end{definition}

\begin{lemma}[Success Event] \label{lem:success-event}
With high probability, \cref{alg:restricted-sssp} is successful.
\end{lemma}
\begin{proof}
Fix any level $\ell > 1$, repetition $r$, and vertex $v$, and let $S \in \mathcal S_{\ell, r}$ be the unique cluster containing $v$. Moreover, fix an arbitrary shortest $S$-$v$-path $P$ in $G[S]$. We show that for any~\makebox{$r' \in [R]$}, with probability at least $\frac{1}{2}$ the path $P$ is cut at most $t$ times in~$\mathcal S_{\ell-1, r'}$ (recall throughout that $S_{\ell-1, r'}$ is sampled independently of $\mathcal S_{\ell, r}$).

The cluster $S$ has weak diameter at most $2^\ell$ in $G_{\geq 0}$ (by \cref{cor:topo-sorted-directed-ldd}). Thus, by the previous lemma we have $w_{\geq 0}(P) \leq 2^\ell$. Since the LDD cuts each edge $e$ with probability at most~\smash{$\frac{w(e)}{d} \cdot Q(n) + \bigO(\frac{1}{n^8})$}, the expected number of edges on $P$ that is cut in $\mathcal S_{\ell-1, r'}$ is at most
\begin{equation*}
    \sum_{e \in P} \parens*{\frac{w_{\geq 0}(e)}{2^{\ell-1}} \cdot Q(n) + \Order(n^{-8})} = \frac{w_{\geq 0}(P)}{2^{\ell-1}} \cdot Q(n) + \Order(n^{-7}) \leq 2 Q(n) + \Order(n^{-7}).
\end{equation*}
In particular, by Markov's inequality with probability at least $\frac{1}{2}$ we cut at most twice that many edges, proving the claim as $t = 4Q(n) + 1$.

Recall now that the LDDs $\mathcal S_{\ell-1, r'}$ are sampled independently across the repetitions $r'$. Thus, the probability that in none of these samples we cut $P$ at most $t$ times is at most $2^{-R}$. Taking a union bound over all levels $\ell$, repetitions $r$ and vertices $v$, the total error probability is indeed at most $\poly(\log n) \cdot n \cdot 2^{-R} \leq n^{-c_{\mathrm{whp}}}$ as we have picked $R = (2 + c_{\mathrm{whp}}) \log n$.
\end{proof}

\begin{lemma}[Correctness] \label{lem:restricted-sssp-correctness}
With high probability, \cref{alg:restricted-sssp} correctly solves the Restricted SSSP problem.
\end{lemma}
\begin{proof}
We condition on the high-probability event that the algorithm is successful. Our goal is to show the following two statements $\mathcal A(\ell)$ and $\mathcal B(\ell)$ for all levels $\ell$:
\begin{itemize}
    \item $\mathcal A(\ell)$: For all $r \in [R]$ and for each $S_i \in \mathcal S_{\ell, r}$, it holds that $\phi^{(i)}(v) = \dist_{G[S_i]}(S_i, v)$.
    \item $\mathcal B(\ell)$: For all $r \in [R]$, it holds that $w_{\phi_{\ell, r}}(e) \geq 0$ for all edges $e$ not cut by $\mathcal S_{\ell, r}$.
\end{itemize}
We prove both statements by induction on $\ell$ as follows.
\begin{itemize}
    \item \emph{Base case: $\mathcal A(1)$.}
    Focus on any repetition and any cluster $S_i \in \mathcal S_{1, r}$ of weak diameter at most $2$. From \cref{lem:rep-paths-short} it follows that any shortest $S_i$-$v$-path $P$ has hop-length at most~$2$, and thus we compute $\boldsymbol{d}_\phi(v) = \dist_{G[S_i]}(S_i, v)$ by calling $\Call{FewNegSSSP}{G[S_i], \phi, t}$ for $t \geq 2$ (irrespective of the potential $\phi$). The claim (a) follows as we pick~\smash{$\phi^{(i)}(v) = \boldsymbol{d}_\phi(v)$} for the all-zero potential $\phi$.

    \item \emph{Inductive case: $\mathcal B(\ell-1)$ implies $\mathcal A(\ell)$.}
    Focus on any repetition $r$ and any cluster $S_i \in \mathcal S_{\ell, r}$. Recall that by the guarantee of \cref{lem:dijkstra-potential} it holds that (i) $\dist_{G[S_i]}(S_i, v) \leq \boldsymbol{d}_\phi(v)$ for all $\phi$. Moreover, conditioning on the event that the algorithm is successful there is some shortest $S_i$-$v$-path~$P$ in $G[S_i]$ and some repetition $r'$ such that $P$ is cut at most~$t$ times in the clustering $\mathcal S_{\ell-1, r'}$. Therefore, by the assumption $\mathcal B(\ell-1)$, reweighting~$G$ with the potential function~\makebox{$\phi_{\ell-1, r'} \in \Phi_{\ell-1}$} makes all but at most $t$ edges in $P$ nonnegative. Thus, \cref{lem:dijkstra-potential} implies that (ii) $\boldsymbol{d}_{\phi_{\ell-1, r'}}(v) \leq w(P) = \dist_{G[S_i]}(S_i, v)$. Combining~(i) and~(ii), we obtain that indeed
    \begin{equation*}
        \phi^{(i)}(v) = \min_{\phi \in \Phi_{\ell-1}} \boldsymbol{d}_\phi(v) = \dist_{G[S_i]}(S_i, v)
    \end{equation*}
    for all nodes $v \in S_i$.
    \item \emph{Inductive case: $\mathcal A(\ell)$ implies $\mathcal B(\ell)$.}
    Focus on any repetition $r$ and any edge $e = (u, v)$ that is not cut in $\mathcal S_{\ell, r}$. If $u$ and $v$ belong to the same cluster $S_i$, then $w_{\phi_{\ell, r}}(e) = w_{\phi^{(i)}}(e) \geq 0$ by Johnson's trick (\cref{obs:johnson}) since~\smash{$\phi^{(i)} = \dist_{G[S_i]}(S_i, v)$} (by assumption $\mathcal A(\ell)$). So suppose that $u$ and~$v$ belong to different clusters, $u \in S_i, v \in S_j$. We assume that $e$ is not cut, so $i < j$. Therefore:
    \begin{align*}
        w_{\phi_{\ell, r}}(e) &= w(e) + (\phi^{(i)}(u) - i \cdot n) - (\phi^{(j)}(v) - j \cdot n) \\
        &= w(e) + \phi^{(i)}(u) - \phi^{(j)}(v) + (j - i) n \\
        &\geq -1 -(n-1) + 0 + n = 0.
    \end{align*}
    In the last step we used that $G$ is restricted and so $w(e) \geq - 1$ and $-(n-1) \leq \phi^{(i)}(v) \leq 0$ (recalling by assumption $\mathcal A(\ell)$ that $\phi^{(i)}(u) = \dist_{G[S_i]}(S_i, u)$).
\end{itemize}
Altogether the induction shows that $B(L)$ holds, i.e., that for all $r$ we have~\makebox{$w_{\phi_{L, r}}(e) \geq 0$} for all edges $e$ not cut by $\mathcal S_{L, r}$. However, at the $L$-th level we may assume that the LDD does not cut any edges in $G$ as any pair of nodes reachable from each other is trivially within distance at most $(n-1) \cdot n \leq n^2$. In particular, it follows that $w_{\phi_{L, 1}}(e) \geq 0$ for all edges $e \in E$.
\end{proof}

\begin{lemma}[Work and Span] \label{lem:restricted-sssp-work}
\cref{alg:restricted-sssp} has work $W_{\geq 0}(n, m) \cdot \poly(\log n)$ and span $S_{\geq 0}(n, m) \cdot \poly(\log n)$.
\end{lemma}
\begin{proof}
The main bottleneck to both work and span is the calls to \textsc{DirLDD} and \textsc{FewNegSSSP}. Each call to \textsc{DirLDD} runs in work $W_{\geq 0}(n, m)$ and span $S_{\geq 0}(n, m)$ up to polylogarithmic factors (by \cref{cor:topo-sorted-directed-ldd}), and there are at most $LR = \poly(\log n)$ such calls. For each iteration of the loop in Line~9, the calls to $\textsc{FewNegSSSP}$ in Line~11 are executed on disjoint induced graphs~$G[S_i]$ (up to an overhead of $L R \cdot \max_\ell |\Phi_\ell| \leq L R^2 = \poly(\log n)$). By the assumption that $W_{\geq 0}(n, m)$ is subadditive (\Cref{eq:work-span-additive}) the total work is thus at most $W_{\geq 0}(n, m)$ and the total span is at most $S_{\geq 0}(n, m)$ (both up to polylogarithmic factors). Any work outside these calls trivially takes work $\tilde\Order(n + m)$ and span $\poly(\log n)$.

As a final detail, we remark that each call to \textsc{DirLDD} as well as to \textsc{FewNegSSSP} is indeed on graphs with polynomially bounded weights (in fact, weights $\set{0, \dots, n}$) and with polynomially bounded potentials (in fact, $|\phi(v)| \leq n^2$ for all $v$).
\end{proof}

%% file: main.bbl
\begin{thebibliography}{10}

\bibitem{AndoniSZ20}
Alexandr Andoni, Clifford Stein, and Peilin Zhong.
\newblock Parallel approximate undirected shortest paths via low hop emulators.
\newblock In Konstantin Makarychev, Yury Makarychev, Madhur Tulsiani, Gautam
  Kamath, and Julia Chuzhoy, editors, {\em Proceedings of the 52nd Annual {ACM}
  Symposium on Theory of Computing ({STOC} 2020)}, pages 322--335. {ACM}, 2020.
\newblock \href {https://doi.org/10.1145/3357713.3384321}
  {\path{doi:10.1145/3357713.3384321}}.

\bibitem{AshvinkumarBCGHJNS24}
Vikrant Ashvinkumar, Aaron Bernstein, Nairen Cao, Christoph Grunau, Bernhard
  Haeupler, Yonggang Jiang, Danupon Nanongkai, and Hsin{-}Hao Su.
\newblock Parallel, distributed, and quantum exact single-source shortest paths
  with negative edge weights.
\newblock In Timothy~M. Chan, Johannes Fischer, John Iacono, and Grzegorz
  Herman, editors, {\em 32nd Annual European Symposium on Algorithms ({ESA}
  2024)}, volume 308 of {\em LIPIcs}, pages 13:1--13:15. Schloss Dagstuhl -
  Leibniz-Zentrum f{\"{u}}r Informatik, 2024.
\newblock URL: \url{https://doi.org/10.4230/LIPIcs.ESA.2024.13}, \href
  {https://doi.org/10.4230/LIPICS.ESA.2024.13}
  {\path{doi:10.4230/LIPICS.ESA.2024.13}}.

\bibitem{BernsteinNW22}
Aaron Bernstein, Danupon Nanongkai, and Christian Wulff{-}Nilsen.
\newblock Negative-weight single-source shortest paths in near-linear time.
\newblock In {\em 63rd Annual Symposium on Foundations of Computer Science
  (FOCS 2022)}, pages 600--611. {IEEE}, 2022.
\newblock \href {https://doi.org/10.1109/FOCS54457.2022.00063}
  {\path{doi:10.1109/FOCS54457.2022.00063}}.

\bibitem{BringmannCF23}
Karl Bringmann, Alejandro Cassis, and Nick Fischer.
\newblock Negative-weight single-source shortest paths in near-linear time: Now
  faster!
\newblock In {\em 64th {IEEE} Annual Symposium on Foundations of Computer
  Science ({FOCS} 2023)}, pages 515--538. {IEEE}, 2023.
\newblock \href {https://doi.org/10.1109/FOCS57990.2023.00038}
  {\path{doi:10.1109/FOCS57990.2023.00038}}.

\bibitem{CaoF23}
Nairen Cao and Jeremy~T. Fineman.
\newblock Parallel exact shortest paths in almost linear work and square root
  depth.
\newblock In Nikhil Bansal and Viswanath Nagarajan, editors, {\em 34th
  {ACM-SIAM} Symposium on Discrete Algorithms ({SODA} 2023)}, pages 4354--4372.
  {SIAM}, 2023.
\newblock URL: \url{https://doi.org/10.1137/1.9781611977554.ch166}, \href
  {https://doi.org/10.1137/1.9781611977554.CH166}
  {\path{doi:10.1137/1.9781611977554.CH166}}.

\bibitem{CaoFR20}
Nairen Cao, Jeremy~T. Fineman, and Katina Russell.
\newblock Efficient construction of directed hopsets and parallel approximate
  shortest paths.
\newblock In Konstantin Makarychev, Yury Makarychev, Madhur Tulsiani, Gautam
  Kamath, and Julia Chuzhoy, editors, {\em Proceedings of the 52nd Annual {ACM}
  Symposium on Theory of Computing ({STOC} 2020)}, pages 336--349. {ACM}, 2020.
\newblock \href {https://doi.org/10.1145/3357713.3384270}
  {\path{doi:10.1145/3357713.3384270}}.

\bibitem{CaoFR22}
Nairen Cao, Jeremy~T. Fineman, and Katina Russell.
\newblock Parallel shortest paths with negative edge weights.
\newblock In Kunal Agrawal and I{-}Ting~Angelina Lee, editors, {\em 34th {ACM}
  Symposium on Parallelism in Algorithms and Architectures ({SPAA} 2022)},
  pages 177--190. {ACM}, 2022.
\newblock \href {https://doi.org/10.1145/3490148.3538583}
  {\path{doi:10.1145/3490148.3538583}}.

\bibitem{ChenKLPGS22}
Li~Chen, Rasmus Kyng, Yang~P. Liu, Richard Peng, Maximilian~Probst Gutenberg,
  and Sushant Sachdeva.
\newblock Maximum flow and minimum-cost flow in almost-linear time.
\newblock In {\em 63rd {IEEE} Annual Symposium on Foundations of Computer
  Science ({FOCS} 2022)}, pages 612--623. {IEEE}, 2022.
\newblock \href {https://doi.org/10.1109/FOCS54457.2022.00064}
  {\path{doi:10.1109/FOCS54457.2022.00064}}.

\bibitem{Cohen00}
Edith Cohen.
\newblock Polylog-time and near-linear work approximation scheme for undirected
  shortest paths.
\newblock {\em J. {ACM}}, 47(1):132--166, 2000.
\newblock \href {https://doi.org/10.1145/331605.331610}
  {\path{doi:10.1145/331605.331610}}.

\bibitem{Dijkstra59}
Edsger~W. Dijkstra.
\newblock A note on two problems in connexion with graphs.
\newblock {\em Numerische Mathematik}, 1:269--271, 1959.
\newblock \href {https://doi.org/10.1007/BF01386390}
  {\path{doi:10.1007/BF01386390}}.

\bibitem{Fineman24}
Jeremy~T. Fineman.
\newblock Single-source shortest paths with negative real weights in
  \emph{{\~{O}}(mn\({}^{\mbox{8/9}}\))} time.
\newblock In Bojan Mohar, Igor Shinkar, and Ryan O'Donnell, editors, {\em 56th
  Annual {ACM} Symposium on Theory of Computing ({STOC} 2024)}, pages 3--14.
  {ACM}, 2024.
\newblock \href {https://doi.org/10.1145/3618260.3649614}
  {\path{doi:10.1145/3618260.3649614}}.

\bibitem{FredmanT84}
Michael~L. Fredman and Robert~Endre Tarjan.
\newblock Fibonacci heaps and their uses in improved network optimization
  algorithms.
\newblock In {\em 25th Annual Symposium on Foundations of Computer Science
  (FOCS 1984)}, pages 338--346. {IEEE} Computer Society, 1984.
\newblock \href {https://doi.org/10.1109/SFCS.1984.715934}
  {\path{doi:10.1109/SFCS.1984.715934}}.

\bibitem{Gabow83}
Harold~N. Gabow.
\newblock Scaling algorithms for network problems.
\newblock In {\em 24th Annual Symposium on Foundations of Computer Science
  (FOCS 1983)}, pages 248--257. {IEEE} Computer Society, 1983.
\newblock \href {https://doi.org/10.1109/SFCS.1983.68}
  {\path{doi:10.1109/SFCS.1983.68}}.

\bibitem{GabowT89}
Harold~N. Gabow and Robert~Endre Tarjan.
\newblock Faster scaling algorithms for network problems.
\newblock {\em {SIAM} J. Comput.}, 18(5):1013--1036, 1989.
\newblock \href {https://doi.org/10.1137/0218069} {\path{doi:10.1137/0218069}}.

\bibitem{GhaffariZ22}
Mohsen Ghaffari and Goran Zuzic.
\newblock Universally-optimal distributed exact min-cut.
\newblock In Alessia Milani and Philipp Woelfel, editors, {\em 41st {ACM}
  Symposium on Principles of Distributed Computing ({PODC} 2022)}, pages
  281--291. {ACM}, 2022.
\newblock \href {https://doi.org/10.1145/3519270.3538429}
  {\path{doi:10.1145/3519270.3538429}}.

\bibitem{Goldberg95}
Andrew~V. Goldberg.
\newblock Scaling algorithms for the shortest paths problem.
\newblock {\em {SIAM} J. Comput.}, 24(3):494--504, 1995.
\newblock \href {https://doi.org/10.1137/S0097539792231179}
  {\path{doi:10.1137/S0097539792231179}}.

\bibitem{HuangJQ24}
Yufan Huang, Peter Jin, and Kent Quanrud.
\newblock Faster single-source shortest paths with negative real weights via
  proper hop distance.
\newblock {\em CoRR}, abs/2407.04872, 2024.
\newblock URL: \url{https://doi.org/10.48550/arXiv.2407.04872}, \href
  {https://doi.org/10.48550/ARXIV.2407.04872}
  {\path{doi:10.48550/ARXIV.2407.04872}}.

\bibitem{Johnson77}
Donald~B. Johnson.
\newblock Efficient algorithms for shortest paths in sparse networks.
\newblock {\em J. {ACM}}, 24(1):1--13, 1977.
\newblock \href {https://doi.org/10.1145/321992.321993}
  {\path{doi:10.1145/321992.321993}}.

\bibitem{Li20}
Jason Li.
\newblock Faster parallel algorithm for approximate shortest path.
\newblock In Konstantin Makarychev, Yury Makarychev, Madhur Tulsiani, Gautam
  Kamath, and Julia Chuzhoy, editors, {\em Proceedings of the 52nd Annual {ACM}
  Symposium on Theory of Computing ({STOC} 2020)}, pages 308--321. {ACM}, 2020.
\newblock \href {https://doi.org/10.1145/3357713.3384268}
  {\path{doi:10.1145/3357713.3384268}}.

\bibitem{RohzhvnGHZL22}
V{\'{a}}clav Rozhon, Christoph Grunau, Bernhard Haeupler, Goran Zuzic, and
  Jason Li.
\newblock Undirected (1+\emph{{\(\epsilon\)}})-shortest paths via
  minor-aggregates: near-optimal deterministic parallel and distributed
  algorithms.
\newblock In Stefano Leonardi and Anupam Gupta, editors, {\em 54th Annual {ACM}
  Symposium on Theory of Computing ({STOC} 2022)}, pages 478--487. {ACM}, 2022.
\newblock \href {https://doi.org/10.1145/3519935.3520074}
  {\path{doi:10.1145/3519935.3520074}}.

\bibitem{RozhovnHMGZ23}
V{\'{a}}clav Rozhon, Bernhard Haeupler, Anders Martinsson, Christoph Grunau,
  and Goran Zuzic.
\newblock Parallel breadth-first search and exact shortest paths and stronger
  notions for approximate distances.
\newblock In Barna Saha and Rocco~A. Servedio, editors, {\em 55th Annual {ACM}
  Symposium on Theory of Computing ({STOC} 2023)}, pages 321--334. {ACM}, 2023.
\newblock \href {https://doi.org/10.1145/3564246.3585235}
  {\path{doi:10.1145/3564246.3585235}}.

\bibitem{Williams64}
John W.~J. Williams.
\newblock Algorithm 232 -- {Heapsort}.
\newblock {\em Communications of the ACM}, 7(6):347--349, 1964.
\newblock \href {https://doi.org/10.1145/512274.512284}
  {\path{doi:10.1145/512274.512284}}.

\bibitem{ZuzicGYHS21}
Goran Zuzic, Gramoz Goranci, Mingquan Ye, Bernhard Haeupler, and Xiaorui Sun.
\newblock Universally-optimal distributed shortest paths and transshipment via
  graph-based l1-oblivious routing.
\newblock {\em CoRR}, abs/2110.15944, 2021.
\newblock URL: \url{https://arxiv.org/abs/2110.15944}, \href
  {https://arxiv.org/abs/2110.15944} {\path{arXiv:2110.15944}}.

\end{thebibliography}
